\documentclass[]{interact}

\usepackage{epstopdf}
\usepackage[caption=false]{subfig}

\usepackage[numbers,sort&compress]{natbib}
\bibpunct[, ]{[}{]}{,}{n}{,}{,}

\theoremstyle{plain}
\newtheorem{theorem}{Theorem}[section]

\theoremstyle{definition}

\theoremstyle{remark}

\begin{document}


\title{Peak Effects in Stable Linear Difference Equations}

\author{
\name{B.~T. Polyak\textsuperscript{a}\thanks{CONTACT: B.~T. Polyak. Email: boris@ipu.ru},
P.~S. Shcherbakov\textsuperscript{a}, and
G. Smirnov\textsuperscript{b}}
\affil{\textsuperscript{a}Institute of Control Science, RAS, 65, Profsoyuznaya, Moscow, Russia;\\
~~e-mail:~~boris{@}ipu.ru, cavour118{@}mail.ru\\
\textsuperscript{b}University of Minho, Braga, Portugal;\\
~~email:~~smirnov{@}math.uminho.pt}
}

\maketitle

\begin{abstract}
We consider asymptotically stable scalar difference equations with unit-norm initial conditions.
First, it is shown that the solution may happen to deviate far away from the equilibrium point at finite time instants prior to converging to zero.
Second, for a number of root distributions and initial conditions, exact values of deviations or lower bounds are provided.
Several specific difference equations known from the literature are also analyzed and estimates of deviations are proposed.
Third, we consider difference equations with non-random noise (i.e., bounded-noise autoregression) and provide upper bounds on the solutions.
Possible generalizations, e.g., to the vector case are discussed and directions for future research are outlined.

\end{abstract}


\section{Introduction}

Analysis of transients in systems described by ordinary differential or difference equations has always been the subject
of intense research in automatic control theory, signal processing, numerical analysis, and other numerous fields.
Usually, by transient is meant the response of a stable system with zero initial conditions to
a typical input signal such as unit-step, harmonic, etc. For instance, in control theory, possible large deviations from the steady state
are referred to as \emph{overshoot}, and there exist many results on these effects in linear dynamical systems.

Much less attention in the literature has been paid to the closely related phenomena, typically known as \emph{peak effects},
induced by nonzero initial conditions in the absence of exogenous input signal.
Clearly, both overshoot and peaks are to be avoided in the design to provide ``smooth'' (or even monotone) transient.
Among the relatively recent works on peak effects in \emph{continuous} time systems, we mention
\cite{BTSmirnov_Automatica,peakAiT}; 
also see the references therein.

To the best of our knowledge, the discrete-time case is unexplored, and the results cannot be directly derived from the
continuous-time ones, since the nature of peaking effects is different. The basic monograph~\cite{Elaydi} provides a deep and detailed
introduction to the general theory of linear and nonlinear difference equations, stability, asymptotic behavior of solutions, their properties,
and numerous applications. A closely related issue, the behavior of norms of powers of Schur stable matrices in the autonomous vector scheme
$x_{k+1} = Ax_k$, $x_k\in{\mathbb R}^n$, is also discussed; however peaking effects are not considered.
For this vector scheme, there exist LMI-based results on upper bounds \cite{kogan,malta}, but the most interesting problem,
lower estimates of the transient behavior, is still kept in shade, both in the vector and scalar cases.

The importance of research of peak effects in linear difference equations is motivated by numerous reasons.
For instance, analysis of nonlinear processes, which are typical to practical applications is often performed via use of their
linearizations in a neighborhood of a stable operating point.
Possible large deviations of the linearized trajectories may lead to leaving the basin of attraction and loss of validity.

Another motivation is the analysis of nonasymptotic behavior of the currently popular methods for function minimization,
such as the heavy-ball method~\cite{heavy_ball} and Nesterov's accelerated gradient descent~\cite{nesterov}.
If applied to quadratic functions, these methods are described by second-order linear vector difference equations.
It was discovered via numerous simulations ~\cite{GB} that these methods exhibit nonmonotonic convergence to the minimum.
Peaking effects can be an explanation of such behavior.

Clearly, for a particular difference equation and specified initial conditions, the peak of the trajectory can be found directly via
numerical simulations. However, for high dimensions, computations may not be numerically stable.
Next, finding the worst-case initial conditions in the unit ball in one or another vector norm leading to the highest peak is not easy;
moreover, estimating the values of peaks for \emph{classes of equations} is much more complicated and most challenging.
Last but not least, from the control theory perspective, design of minimum-peak feedback is extremely important.
Overall, exploration in this direction is highly demanded.

Another related field of research is the examination of peaking effects caused by both the initial conditions and exogenous \textit{noise}.
In contrast to the assumption on the Gaussian nature of noise typically adopted in such autoregression models,
we study the situation with bounded non-random perturbations. Such models are in common use in the population dynamics and macroeconomics,
e.g., see~\cite{Turchin} and ~\cite{OS}, and the analysis of possible peak phenomena observed in these models is in demand.

In this paper we focus our attention on some of these issues.
We provide examples where peaks of solutions of stable scalar difference equations are unavoidable and estimate the value of peak
in certain specific cases, both homogenous and nonhomogenous; sometimes, exact closed-form expressions are obtained.

The following notation is used in the paper:
$\mathbb R$ stands for the field of real numbers;
$\mbox{}^\top$~is the transposition sign;
$:=$ denotes equality by definition;
$|\cdot|$ denotes the absolute value of a number, and
$\|\cdot \|$ is a norm of a vector or a matrix;
${\rm e}$ stands for ${\rm exp}(1)$;
for integer $p\geq q$, the binomial coefficient is denoted by $C_p^q = \tfrac{p!}{q!(p-q)!}$;
$\lfloor\cdot\rfloor$ and $\lceil\cdot\rceil$ denote rounding toward negative and positive infinity, respectively.

We also note that the first preliminary results in this area of research were presented by the authors in the conference papers~\cite{timisoara,malta}.

\section{Large deviations of solutions}
\label{s:generic}

\subsection{Formulation of the problem}
We consider the generic $n$th order scalar linear difference equation of the form
\begin{equation}
\label{eq:generic}
x_k + a_1x_{k-1} + \dots + a_nx_{k-n} = 0, \quad k = n,n+1,\dots; \quad a_i\in{\mathbb R},
\end{equation}
with initial conditions
\begin{equation*}
\label{generic:init_cond}
x^{(0)}=(x_{0}, \;\dots,\; x_{n-1}) \in {\mathbb R}^n.
\end{equation*}
The characteristic polynomial of~\eqref{eq:generic} is
\begin{equation}
\label{eq:poly}
p(\lambda) = \lambda^n + a_1\lambda^{n-1} + \dots + a_{n-1}\lambda + a_n,
\end{equation}
and it is assumed to be stable, i.e., all its roots $\lambda_i$ belong to the open unit disc on the complex plane, so that for any finite initial $x^{(0)}$,
the solution asymptotically tends to zero as $k\to\infty$. Without loss of generality we assume that $x^{(0)}$ has unit norm,
namely, $\|x^{(0)}\|_\infty = 1$.
Our aim is to characterize the \emph{nonasymptotic} behavior of solutions; specifically, we are interested in estimating the following quantity:
\begin{equation}
\label{peak001}
\eta( x^{(0)}) = \max_{k=n,n+1,\dots} |x_k|,
\end{equation}
which will be referred to as \emph{peak} of the solution (provided that $\eta( x^{(0)})>1$) for a given root location~$\Lambda$ of equation~\eqref{eq:generic}
and given initial condition~$x^{(0)}$. As said, in principle, finding~$\eta( x^{(0)})$ can be performed via straightforward computations.
A more interesting problem is to estimate the quantity (\emph{upper bound on peak})
$$
\eta = \max_{\Lambda}\max_{\|x^{(0)}\|_\infty=1}\eta(x^{(0)})
$$
which relates to the \emph{worst-case} initial conditions in the unit box for some class of root locations~$\Lambda$, or
$$
\underline{\eta}(x^{(0)}) = \min_{\Lambda}\eta(x^{(0)}),
$$
which is a lower bound on peak for a class of root locations and given initial conditions.

For the same difference equation we also consider the nonhomogenous case:
\begin{equation}
\label{eq:autoreg}
x_k + a_1x_{k-1} + \dots + a_nx_{k-n}=v_k, \quad |v_k| \le \varepsilon, \; v_k \in{\mathbb R}, \quad k = n,n+1,\dots
\end{equation}
That is, we deal with the \textit{autoregression model}; however, in contrast to the standard framework where \textit{Gaussian noise}
is considered, we analyse \textit{unknown-but-bounded noise}. The goal is to estimate upper bounds on $x_k$ in this situation.

\subsection{Equal roots: A closed-form solution}
\label{SS:equal_roots}

In this section we obtain closed-form expressions for the values of peak and peak instant for the solutions of~\eqref{eq:generic}
in the case where all roots of~\eqref{eq:poly} are equal, and demonstrate the possibility of large values of peak.

Consider equation~\eqref{eq:generic} having all roots  $\lambda_i$ of~\eqref{eq:poly} real and equal to $ \rho\in(0,1)$,
and initial conditions $x^{(0)}=(x_0,\; x_1,\;,\dots, x_{n-1})$.  Then \eqref{eq:generic} has the form
\begin{equation}
\label{eq:polyeq}
x_k -C_n^1\rho x_{k-1} + C_n^2\rho^2 x_{k-2}\dots + (-1)^n\rho^n x_{k-n} = 0.
\end{equation}
As is well known \cite{Elaydi} the solution of this difference equation is
$$
x_k=P(k)\rho^k,
$$
where the coefficients of the $(n\!-\!1)$st order polynomial $P(k)$ can be found from the initial conditions $x^{(0)}$.
We represent $P(k)$ in the Lagrange interpolation form as
$$
P(k)=\sum_{i=0}^{n-1} c_i P_i(k), \quad P_i(k)= \prod_{j\neq i}\frac{k-j}{i-j}, \quad i=0,1,\dots, n-1, \;\; j=0,1,\dots, n-1.
$$
We then have $P_i(i)=1$, $P_i(k)=0$ for $i=0,1,\dots, n-1, \; k\neq i, \; k=0,1,\dots, n-1$.
Also, for a fixed $k\geq n$, the signs of $P_i(k)$ are seen to alternate:
$$
P_{n-1}(k)>0, \quad P_{n-2}(k)<0, \;\dots, \quad {\rm sgn}\, P_0(k)=(-1)^{n-1}.
$$
The coefficients $c_i$ can be found from the initial conditions (we use the interpolation property of the polynomials $P_i$):
$$
x_0=c_0, \quad x_1=c_1\rho, \;\dots, \quad x_{n-1}=c_{n-1}\rho^{n-1}.
$$
Hence, we arrive at the closed-form solution of difference equations with all equal roots:
\begin{equation}
\label{equal}
x_k=x_0P_0(k)\rho^k+x_1P_1(k)\rho^{k-1}+\dots x_{n-1}P_{n-1}(k)\rho^{k-n+1}.
\end{equation}

In particular, for $n=2$ we have
$$
x_k=-x_0(k-1)\rho^k+x_1k\rho^{k-1},
$$
while for $n=3$ the explicit expression writes
$$
x_k=x_0\frac{(k-1)(k-2)}{2}\rho^k-x_1k(k-2)\rho^{k-1}+x_2\frac{k(k-1)}{2}\rho^{k-2}.
$$

As far as the peak value of $x_k$ is concerned, it is seen from~\eqref{equal} and definition~\eqref{peak001}
that $\eta(x^{(0)}) = \eta(-x^{(0)})$, since change of sign of the initial conditions changes just the sign of the solution.

To continue, due to the sign alternating property of the polynomials $P_i$ we conclude that for all $k\geq n$, the following relation holds:
$$
\max_{\|x^{(0)}\|_\infty = 1} x_k=|P_0(k)|\rho^k+|P_1(k)|\rho^{k-1}+\dots +|P_{n-1}(k)|\rho^{k-n+1}\, :=\,\alpha_{k,n},
$$
and this maximum is attained with $x^{(0)}=((-1)^{n-1}, \,\dots, -1,\, 1)$.

For another initial conditions $x^{(0)}=(0, \dots, 0, 1)$ we have
$$
x_k= P_{n-1}(k) \rho^{k-n+1}=C_k^{n-1}\rho^{k-n+1}\,:=\,\beta_{k,n};
$$
obviously $\alpha_{k,n}>\beta_{k,n}$.

Denote now
$$
\alpha_{n}=\max_k\alpha_{k,n}, \;\; K_\alpha = {\rm argmax} \,\alpha_{k,n}
\mbox{~~~~and~~~~}
\beta_{n}=\max_k\beta_{k,n}, \;\; K_\beta = {\rm argmax} \,\beta_{k,n}.
$$

By optimizing the expression for~$\beta_{k,n}$ over~$k$, we immediately obtain the exact expression
$$
K_{\beta} = \left\lfloor \frac{n-1}{1-\rho}\right\rfloor.
$$
By definition, peak takes place ($\beta_n>1$) if and only if $K_\beta>n-1$, which is equivalent to $\rho> \tfrac{1}{n}$.

We arrive at the following statement.
\begin{theorem}
\label{th:equal_real_roots}
For arbitrary initial conditions $x^{(0)}$, the closed-form solution of the difference equation~\eqref{eq:polyeq} is given by \eqref{equal}.
The upper bound on $\eta(x^{(0)})$ for all $\|x^{(0)}\|_{\infty} \leq 1$ is given by
$$
x_k\leq \alpha_{k,n}, \quad k=n, n+1, \dots ,
$$
and it is attained with $ x^{(0)}=((-1)^{n-1}, \dots, -1, 1)$ and $k=K_{\alpha}$.

For the initial conditions $x^{(0)}=(0,\dots ,0,1)$, the solution of~\eqref{eq:polyeq} is given by
$$
x_k =  \beta_{k,n}=C_k^{n-1}\rho^{k-n+1}, \qquad k=n, n+1, \dots ,
$$
and the maximum is attained with $k=K_{\beta} = \lfloor \tfrac{n-1}{1-\rho}\rfloor$. The peak takes place only for $\rho>\rho^* = \tfrac{1}{n}$.
\end{theorem}

Several comments are due a this point.

First, for a given order~$n$ of the equation, there always exist initial conditions such that the magnitude of peak grows as $\rho\to 1$;
same for the magnitude of the peak instant.

Second, for large values of~$n$, there exist initial conditions (e.g., $x^{(0)}=(0,\dots ,0,1)$) that yield peaks for small values of~$\rho$
(see the expression for~$\beta_{k,n}$ above); i.e., ``very stable'' higher-order difference equations may experience huge peaks.

Third, in some cases, the expressions for the peak and peak instant for small~$n$ and~$\rho$ can be given in closed form;
sometimes we present approximate estimates or asymptotic expressions as $\rho\to 1$.

A typical shape of the trajectory $x_k$ is given in Fig.~\ref{fig:traj_peak} for the two initial conditions discussed above.
\begin{figure}[h!]
	\centering
	\includegraphics[trim=0 0 0 0,clip,width=3.1in]{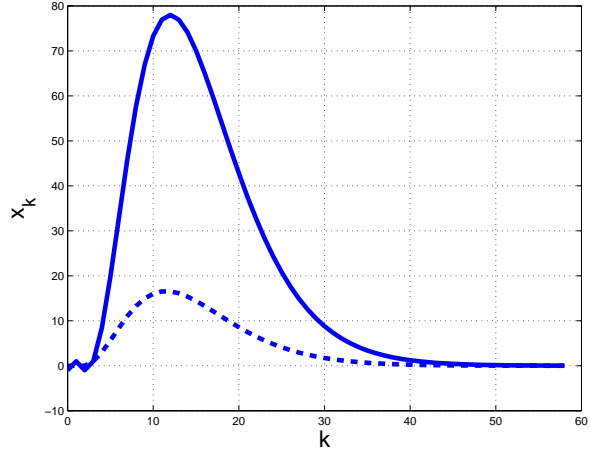}
	 \caption{The values of $x_k$~\eqref{eq:generic} for $n=4, \lambda_i\equiv\rho=0.75$, and $x^{(0)}=(0,\, 0,\, 0,\, 1)$ (dashed)
              or $x^{(0)}=(-1,\, 1,\, -1,\, 1)$ (solid)}
   \label{fig:traj_peak}
\end{figure}

Let us now specify the results of Theorem~\ref{th:equal_real_roots} for some particular cases
(below, the sign $\approx$ is used to mean asymptotics as $\rho\rightarrow 1$).
The values of $\alpha_n, \beta_n$ are found by differentiating $\alpha_{k,n}$, $\beta_{k,n}$ in $k$.

\begin{itemize}
\item
The case $n=2$.

For the initial conditions $x^{(0)}=(-1, 1)$ we have
$$
\alpha_{k,2}=(k-1)\rho^k+k\rho^{k-1}, \quad
K_{\alpha}=\frac{\rho(\log\rho-1)}{(1+\rho)\log\rho}\approx \dfrac{1}{1-\rho}, \quad
\alpha_2\approx \frac{2}{{\rm e}(1-\rho)},
$$
and peak is observed only for $\rho>\rho^* = \sqrt{2}-1$.

For the initial conditions $x^{(0)}=(0, 1)$ we have
$$
\beta_{k,2}=k\rho^{k-1}, \quad K_{\beta}=\frac{1}{-\log\rho} \approx\frac{1}{1-\rho} , \quad \beta_2\approx \frac{1}{{\rm e}(1-\rho)}.
$$

\item
The case $n=3$.

For the worst-case initial $x^{(0)}=(1, -1, 1)$ we obtain
$$
\alpha_{k,3}=\frac{(k\!-\!1)(k\!-\!2)}{2}\rho^k\!+k(k\!-\!2)\rho^{k\!-\!1}\!+\frac{k(k\!-\!1)}{2}\rho^{k\!-\!2}, \;\;
K_{\alpha}\!\approx\! \frac{2}{1\!-\!\rho}, \;\;
\alpha_3\!\approx\! \frac{8}{{\rm e}^2(1\!-\!\rho)^2},
$$
and peak takes place only for $\rho>\rho^*=0.2599$, $\rho^*$ being the largest root of the equation $\rho^3+3\rho^2+3\rho -1=0$.

For $x^{(0)}=(0, 0, 1)$ we have
$$
\beta_{k,3}=\frac{k(k-1)}{2}\rho^{k-2}, \quad K_{\beta}\approx \frac{2}{1-\rho}, \quad \beta_3\approx \frac{2}{{\rm e}^2(1-\rho)^2} .
$$

\item
Consider $n$ fixed and large and $\rho\rightarrow 1$. Then for $k/n$ large we have
$$
\beta_{k,n}\approx \frac{k^n}{n!}\rho^{n-k}, \quad \alpha_{k,n}\approx 2^{n-1} \beta_{k,n}.
$$
\end{itemize}

As it was already stressed, Theorem~\ref{th:equal_real_roots} shows that  the values of both peak and peak instant grow as $\rho\to 1$ and/or $n$ grow;
this is in a certain analogy with the continuous-time case, see  \cite{peakAiT}, Section 3.1.
We  illustrate the cumulative effect of increasing~$\rho$ and~$n$ by the special case $\rho=1-\tfrac{1}{n}$ and $x^{(0)}=(0,\dots , 0, 1)$.
Having $K_\beta = n^2-n-1$ (see Theorem~\ref{th:equal_real_roots}), we obtain the exact closed-form expression for the peak
$$
\beta_n= C_{n^2-n-1}^{n^2-2n}(1-\tfrac{1}{n})^{n^2-2n}.
$$
For some values of~$n$, Table~1 presents the values of $\beta_n, K_\beta$,
together with the values of $\alpha_n, K_\alpha$, which are obtained numerically.

\begin{table}[h!]
{\bf ~~~~Table 1.} The dependence of $\alpha_n, \beta_n$, and $K_\alpha, K_\beta$ on~$n$ for $\rho = 1-\tfrac{1}{n}$
\begin{center}
\begin{tabular}{|c|c|c|c|c|c|c|c|c|c|c|c|c|c|}
\hline
$n$            &    2   &     3    &     4    &           5          &           6          &         7           \\
\hline
$\beta_n$      &    1   &  2.9630  &  16.519  &         136.37       &  1.4938$\cdot 10^3$  &  2.0405$\cdot 10^4$ \\
\hline
$\alpha_n$     &  1.25  &  7.0014  &  78.002  &  1.2925$\cdot 10^3$  &  2.8408$\cdot 10^4$  &  7.7812$\cdot 10^5$ \\
\hline
$K_{\beta}$    &   1,2  &    5,6   &  11,12   &         19,20        &         29,30        &       41,42         \\
\hline
$K_{\alpha}$   &    2   &     6    &    12    &           20         &           30         &         42          \\
\hline
\end{tabular}
\end{center}
\end{table}

We now present several simple conditions for the existence of peak.
First, an obvious sufficient condition is $x_n>1$. The next one is formulated in terms of the coefficients of~\eqref{eq:generic};
namely, the solution experiences peak if and only if $\sum_i|a_i|>1$.
Next, for the initial $x^{(0)}=(0, \dots, 0,\, 1)$, a sufficient condition for the existence of peak is $|\sum_i \lambda_i|>1$,
and if only positive roots are considered, this conditions also becomes necessary.

Clearly, peak effects strongly depend on $x^{(0)}$. Above, we were focused at either the worst-case initial conditions or at $x^{(0)}=(0, \dots, 0,\, 1)$.
However there are many initial conditions where this effect is lacking.
For instance with $x^{(0)}=(1, \rho,\, \rho^2,\dots , \rho^{n-1})$, it can be easily checked that $x_k=\rho^k$,
hence the solutions are monotonically decreasing.
The only ``$\rho$-independent'' initial conditions leading to the absence of peak for any $\rho \in [0,1)$ is $x^{(0)}=(1, \dots,\, 1)$.

The proofs of all these assertions are immediate.

\subsection{Real roots:  Lower and upper bounds on peak}
\label{SS:lower_upper_bounds}

We next present an important result that can be thought of as bounds of peak for real roots.

Consider equation \eqref{eq:generic} with all real roots of the characteristic polynomial
$$
p(\lambda)=(\lambda-\lambda_1)(\lambda-\lambda_2)\dots (\lambda-\lambda_n).
$$
We then have
$$
p(\lambda)=\lambda^n-S_1(\lambda_1, \dots ,\lambda_n)\lambda^{n-1}+ \dots + (-1)^n S_n(\lambda_1, \dots ,\lambda_n),
$$
where $S_k(\lambda_1, \dots ,\lambda_n)$ is a homogeneous symmetric form of order $k$.
Comparing this expression with \eqref{eq:polyeq} we conclude that for
$\lambda_i\geq \rho>0$, $i=1, \dots, n$,
the coefficients of the difference equation \eqref{eq:generic} satisfy the inequalities
$$
-a_1\geq C_n^1 \rho^{n-1}, \qquad a_2\geq C_n^2 \rho^{n-2},\;\dots
$$
These necessary conditions for the inequalities $\lambda_i\geq \rho$, $i=1, \dots ,n$, to hold will be exploited below.

\begin{theorem}
\label{th:bounds_real}
Consider equation \eqref{eq:generic}  and assume that all roots $\lambda_i$ of \eqref{eq:poly} are real.

(a) For any difference equation with $\lambda_i\geq \rho>0$, $i=1, \dots ,n$,
there exists an initial value $x^{(0)}, \|x^{(0)}\|_{\infty}= 1$ such that for all $k\geq n$ we have
$$
x_k\geq \beta_{k,n}.
$$
The value of peak (if any) is lower bounded by $\beta_n$:
$$
\eta(x^{(0)}) \geq \beta_n.
$$

(b) If $|\lambda_i |\leq \rho< 1$, then for the initial condition $x^{(0)}=(0, \dots ,0, 1)$ the solution is upper bounded by
$$
|x_k|\leq \beta_{k,n}
$$
for all $k\geq n$. Thus the peak value (if any) is upper bounded by $\beta_n$:
$$
\eta(x^{(0)})\leq \beta_{n}.
$$
\end{theorem}

\begin{proof}
Assuming the roots  $\lambda_i$ of \eqref{eq:poly} are distinct we have the following explicit expression for $x_k$:
$$
x_k=c_1\lambda_1^{k}+c_2\lambda_2^{k}+\dots +c_n\lambda_n^{k},
$$
where $c_i$ can be found from the  initial conditions $x_{0}=x_{1}=\dots =x_{n-2}=0, x_{n-1}=1$.
Specifically, for $c\in \mathbb R^n$ we have the linear equation $Vc=e_n$, where $V$ is the standard \textit{Vandermonde matrix}
and $e_n=(0, \dots , 0, 1)^\top$. It implies $c=V^{-1}e_n=g$ with~$g$ being the last column of $V^{-1}$.
The formula for $g=(g_1, \dots g_n)$ is well known (e.g., see \cite{MS}):
$$
g_i=\frac{1}{\Pi_{j\neq i}(\lambda_i-\lambda_j)};
$$
hence
$$
x_k=\sum_i\frac{\lambda_i^k}{\prod_{j\neq i}(\lambda_i-\lambda_j)}.
$$
We now remind the reader of the following result in \cite{PS}, Part 5, Problem 97.
For a differentiable function $f(\lambda)$ and a set of points $\lambda_1, \dots , \lambda_n\in \mathbb R$ satisfying
$a\leq \lambda_i\leq b$ for all $1\leq i\leq n$, there exists $a\leq \tau \leq b$ such that
$$
\frac{f^{(n-1)}(\tau)}{(n-1)! }=\sum_i\frac{f(\lambda_i)}{\prod_{j\neq i}(\lambda_i-\lambda_j)}\,,
$$
where $f^{(n-1)}(\cdot)$ denotes the $(n-1)$st derivative of $f(\cdot)$.
Taking $f(\lambda)=\lambda^k$ we obtain
$$
x_k=\frac {k (k-1)\dots (k-n+2)\tau^{k-n+1}}{(n-1)!} = C_{k}^{n-1}\tau^{k-n+1}\,.
$$

In  case (a) we have $ a=\rho$, $b=1$, so that
$$
x_k = C_{k}^{n-1}\tau^{k-n+1}, \quad \rho\leq \tau\leq 1.
$$
Comparing this with the value of~$\beta_{k,n}$ obtained in Theorem~\ref{th:equal_real_roots} for $x^{(0)} = (0,\dots ,0,1)$
 we arrive at the desired result.

In case (b) we have $a=-\rho$, $b=\rho$ and
$$
x_k=C_k^{n-1}\tau^{k-n+1}, \quad |\tau|\leq \rho.
$$
This quantity  is no greater than the one obtained for all roots equal to~$\rho$.

We can get rid of the initial assumption (all roots are distinct)  via continuity arguments.
\end{proof}

For real root locations, Theorem~\ref{th:bounds_real} shows the importance of the case of equal roots; it provides upper and lower bounds
for the peak. Below, a more general result is formulated, though just as a conjecture, since so far, we can neither prove it, nor find a counterexample.

\vskip .1in
{\bf Worst-case Conjecture.}
Consider equation~\eqref{eq:generic}.
The maximum value of~$|x_k|$ over all root locations
on the disc of radius~$0<\rho<1$ and over all initial conditions~$x^{(0)}$ in the unit cube is attained with $\lambda_i\equiv \rho$
and $x^{(0)}=((-1)^{n-1}, \dots, -1, 1)$.

\subsection{Bounded-error autoregression}

We now proceed to the nonhomogenous equation~\eqref{eq:autoreg}.
Given a finite time horizon $t\ge n$, our goal is to estimate $\max x_t$ subject to all admissible initial conditions and all bounded disturbances~$v_k$.
We arrive at the problem which can be easily transformed into the \textit{linear program}
\begin{eqnarray}
\label{eq:LP}
&\max x_t \mbox{~~subject to}\\
\nonumber
&\|x^0\|_{\infty}\le 1, \quad |v_k|\le \varepsilon,\quad
x_k + a_1x_{k-1} + \dots + a_nx_{k-n}=v_k, \quad k=n,\dots , t.
\end{eqnarray}
Hence, for a fixed equation and fixed $t$ such estimate can be easily found numerically.
However, to highlight the dependence on the roots of the characteristic polynomial, we present some theoretical results.
We restrict our analysis with the most interesting case by assuming (as in Subsection~\ref{SS:equal_roots}) that \textit{all roots are equal}.

\begin{theorem}
\label{th:noises}
Consider equation \eqref{eq:autoreg}  and assume that all roots $\lambda_i$ of \eqref{eq:poly} are equal to~$\rho$, $0<\rho<1$,
and the noise is bounded: $|v_k|\leq \varepsilon$ for $k\ge n$.
	
Then the solution of the optimization problem \eqref{eq:LP} for any $t\geq n$ is given by $ x^0=((-1)^{n-1}, \dots, -1, 1)$
and $v_k=\varepsilon$, $k=n,\dots, t$; the following estimate holds:
$$
   \max x_t \,\le \,\alpha_{t,n} +\varepsilon \sum_{k=n}^{t}C_{t-k+n-1}^{n-1}\rho^{t-k} \,<\, \alpha_{n}+\varepsilon (1-\rho)^{-n}.
$$
\end{theorem}

\begin{proof}
If $x_0=x_1=\dots x_{k-1}=0$, $v_k\neq 0, \, v_i=0, \, i>k$,
then in accordance with the closed-form solution of the noise-free difference equation
with  equal roots~\eqref{eq:polyeq}, after change of notation we obtain
\begin{equation*}
  x(t) = v_k C_{t-k+n-1}^{n-1}\rho^{t-k}.
\end{equation*}

By the general formula for the solution of non-homogeneous difference equations with arbitrary initial conditions and arbitrary
exogenous noise $v_k$ we have
\begin{equation*}
	x(t)=\bar{x}_t + \tilde x_t,\quad  \tilde x_t^{\mbox{}} = \sum_{k=n}^{t}v_k C_{t-k+n-1}^{n-1}\rho^{t-k},
\end{equation*}
where $\bar{x}_t$ denotes the solution of the homogenous equation, see \eqref{eq:polyeq}.
Optimization over $x^0$ has been performed in Theorem 2.1; the optimal initial conditions were shown to be $((-1)^{n-1}, \dots, -1, 1)$.
Since the coefficient at the $v_k$ is positive, the maximizer of the $\tilde x_t$ is $v_k=\varepsilon$.
Next,
\begin{equation*}
 \sum_{k=n}^{t} C_{t-k+n-1}^{n-1}\rho^{t-k} <  \sum_{i=n-1}^{\infty} C_i^{n-1}\rho^{i-n+1} := S_n,
\end{equation*}
and it is not hard to obtain the recursive relation $S_n - \rho S_n=S_{n-1}$.
Since $S_0=1$, we arrive at $S_n=(1-\rho)^{-n}$.	
\end{proof}

We conclude that deviations of solutions from zero are caused by the two reasons:
(i)~peak effect for the homogeneous equation, induced by nonzero initial conditions
and (ii)~monotone transition process due to noise (clearly, the limiting value $x^*$ of the solution for $v_k=\varepsilon, \; k=n, \dots$,
satisfies the equation $x^*(1-\rho)^n=\varepsilon$; moreover, this holds for \emph{arbitrary} initial conditions).
For the same initial conditions, the resulting shapes of trajectories differ depending on the level of noise.
Figure~\ref{fig:noise} depicts the behavior of solutions of the $4$th order equation with equal roots $\lambda_i=0.75$, worst-case
initial $x^{(0)}=(-1,\, 1,\, -1,\, 1)$, and various values $\varepsilon$ of the noise level (cf. Fig.~\ref{fig:traj_peak}).
\begin{figure}[h!]
  \centering
    \includegraphics[trim=0 0 0 0,clip,width=3in]{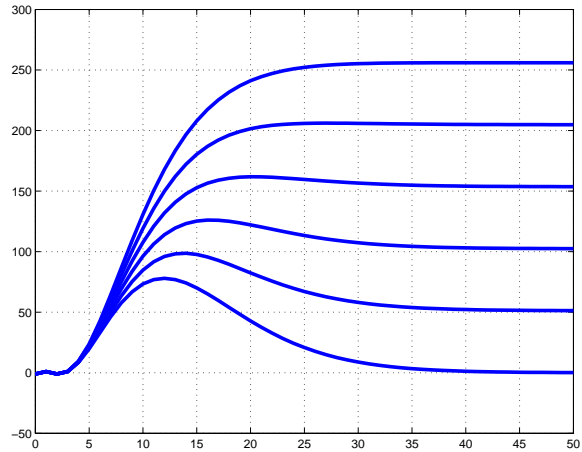}
    \caption{The values of $x_k$~\eqref{eq:generic} for $n=4$, $\lambda_i\equiv\rho=0.75$, and $x^{(0)}=(-1,\, 1,\, -1,\, 1)$ with various
    levels of noise: $\varepsilon = 0; 0.2; 0.4; 0.6; 0.8; 1.0$ (bottom to top).}
  \label{fig:noise}
\end{figure}
It is worth mentioning that for other types of difference equations (e.g., with complex, not real roots)
the solution of the optimization problem \eqref{eq:LP} can be obtained for noise with different signs.

\section{Two special equations}

In this section we analyze peak phenomena for the two input-free difference equations known from the literature.

\subsection{Markov's example}
\label{ss:Markov}

We consider the fourth-order equation borrowed from~\cite{Markov}, Chapter 6, Section 27:
\begin{equation}
\label{eq:markov}
x_{k+4} + 2\rho x_{k+3} + 3\rho^2 x_{k+2} + 2\rho^3 x_{k+1} + \rho^4 x_{k} = 0,
\end{equation}
where, $0<\rho<1$ (in~\cite{Markov}, $\rho=1$ was considered) with the characteristic polynomial
$$
p(\lambda) = \lambda^4 + 2\rho\lambda^3 + 3\rho^2\lambda^2 + 2\rho^3\lambda + \rho^4
$$
having the complex roots
$$
\lambda = \rho(-\cos \pi/3 \pm j\sin \pi/3)
$$
of multiplicity two with absolute values equal to $\rho$.  For the initial condition
$$
x_0 = x_1 = x_3 = 0,\; x_2=-1,
$$
the solution of \eqref{eq:markov}
 is easily shown to be
$$
x_{k} = \frac{2\rho^{k-2}(k-1)}{\sqrt{3}} \sin 2\pi k/3 =
\left\{
\begin{array}{cl}
\pm\rho^{k-2}(k-1)  & \mbox{~~~for~} k \neq 3m,  \\
0 & \mbox{~~~otherwise},
\end{array}
           \right.
$$
so that $x_k\to 0$ as $k\to\infty$.
A straightforward analysis similar to the one performed in Section~\ref{SS:equal_roots} gives
$$
K \approx \frac{1}{1-\rho}
$$
for the peak instant and an asymptotic (for the values of~$\rho$ close to 1) formula for the magnitude of peak:
$$
\eta \approx\frac{1}{{\rm e}\rho(1-\rho)}\,.
$$

For instance, with $\rho=0.99$, these estimates give $K = 100$ and $\eta = 37.1595$, while the actual values obtained from
numerical simulations are equal to~$100$ and $36.9730$, respectively.
Both values are seen to grow to infinity as $\rho$ tends to unity.
It can also be shown that peak is observed only for $\rho > \rho^* = 1/\sqrt{3}$.

This example confirms the presence of the peak effect for certain equations with complex roots.

\subsection{Trinomial equations}

For difference equations with three terms, the analysis can be performed in the space of coefficients, for real and complex roots simultaneously.

In this section we consider the following equation of order $n+1$:
\begin{equation}
\label{eq:kuruklis}
x_{k+1} - a x_k + bx_{k-n} = 0,\quad k = n+1, n+2,\dots,
\end{equation}
with nonzero initial conditions $x^{(0)} = (x_0,\, x_1,\dots, x_n)$,
where $a,b\in {\mathbb R}$ are parameters; so the characteristic polynomial is
$$
p(\lambda) = \lambda^{n+1}-a\lambda^n + b.
$$
Equation~\eqref{eq:kuruklis} has been first analyzed in~\cite{Kuruklis} and later was the subject of intense study in~\cite{kipnis_AiT};
also see~\cite{Elaydi} for a discussion and generalizations.

In fact, most of the efforts in the literature have been put to efficient computation of the boundary of the stability domain of~\eqref{eq:kuruklis}
on the parameter plane $\{a,b\}$, not to the analysis of possible peaks. Perhaps the easiest and transparent method for the description of the
boundary was proposed in~\cite{kipnis_AiT} via use of the $D$-decomposition technique.
By way of illustration, the stability domain $\mathcal S$ for~\eqref{eq:kuruklis} with $n=3$ is depicted in Fig.~\ref{fig:stabdom} below.
\begin{figure}[h!]
  \centering
    \includegraphics[trim=0 0 0 0,clip,width=2.5in]{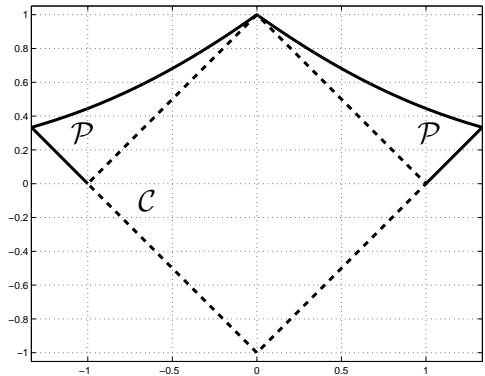}
    \caption{Stability domain and peak domain on the $\{a, b\}$ plane for equation~\eqref{eq:kuruklis} with $n=3$.}
  \label{fig:stabdom}
\end{figure}

Prior to presenting results on peak for various values of the coefficients and initial conditions, we briefly discuss the very possibility of peak.
First, for every~$n$, the domain~$\mathcal S$ contains the so-called Cohn domain, the rhombus ${\mathcal C} = \{a,b\colon |a|+|b|\leq 1\}$ which provides
simple necessary conditions for the stability of~\eqref{eq:kuruklis}. Obviously, for $(a,b)\in {\mathcal S}$, there is no peak for any
initial conditions (cf. end of Section~\ref{SS:equal_roots}),
and the only ``peak-dangerous'' domain on the coefficient plane is ${\mathcal P}$.

From the formulae for the boundary of~$\mathcal S$ (e.g., see~\cite{Kuruklis} or~\cite{Elaydi}, p.249),
a simple upper bound ${\bf A}(\mathcal P)>\tfrac{2}{n}$
on the area of the set ${\mathcal P}$ is immediate to obtain; it is seen to decrease as the order of the equation grows.
For instance, ${\bf A}(\mathcal P) = 0.5 {\bf A}(\mathcal S)$ for $n=1$, and it constitutes less than $4\%$ of $\mathcal S$ for $n=7$.
Though peak effects in equation~\eqref{eq:kuruklis} are seen to be exotic for large dimensions, 
we show below that the magnitude of peak may be arbitrarily large.

We turn to the computation of the magnitude of peak for some specific values of the coefficients and initial conditions.

Let us first consider the following values of the coefficients $(a,b)\in {\mathcal P}$:
$$
1<a\leq1+\frac{\alpha}{n}, \quad 0<\alpha<1, \qquad b := b_2 = a^{n+1}\frac{n^n}{(n+1)^{n+1}}.
$$
We do not present a proof of the feasibility of these coefficients; this follows immediately from the equations of the boundary of~$\mathcal P$.
Instead, we note that the roots of such an equation have the following properties.
The maximal in absolute value root is real, and it has multiplicity two.
By taking the derivative of~$p(\lambda)$ and equating it to zero, for the value of this root we obtain
\begin{equation}
\label{double_root}
\lambda_1 = \lambda_2 := \rho = \frac{an}{n+1}.
\end{equation}
For~$n$ odd, the rest of the roots are complex; for~$n$ even, there is another real (negative) root, which is the least in absolute value.

Let us now consider the initial conditions  of the form
$$
x^{(0)} = (0,\;\rho,\; 2\rho^2, \dots, n\rho^n),
$$
where~$\rho$ is given by~\eqref{double_root}. In contrast to the exposition above, we have $\|x^{(0)}\|_\infty = n\rho^n>\nolinebreak 1$;
hence, by peak we mean
$$
\eta(x^{(0)}) = \frac{1}{n\rho^n}\max_{k\ge n}|x_k|.
$$

These specific initial conditions allow for the exact closed-form expression for the peak and peak instant. Indeed, by induction,
it is immediate to obtain the solution:
$$
x_k = k\rho^k, \qquad k = n+1,n+2,\dots,
$$
and the peak instant~$K$ is given by
$$
K = \left\lfloor\frac{1}{1-\rho}\right\rfloor,
$$
which is obtained by finding the maximum value of~$k$ such that $x_k<x_{k+1}$.
By definition, the peak is observed only if $K > n$; i.e., if $\tfrac{1}{1-\rho}>n+1$, equivalently, for $\rho>\tfrac{n}{n+1}$, i.e., always
(see~\eqref{double_root}).

Respectively, the magnitude of peak is equal to
$$
\eta(x^{(0)})  \,=\, \frac{1}{n\rho^n}\,\left\lfloor\frac{1}{1-\rho}\right\rfloor\rho^{\left\lfloor\tfrac{1}{1-\rho}\right\rfloor}
               \,>\, \frac{1}{n\rho^n}\,\frac{1}{1-\rho}\rho^{\tfrac{1}{1-\rho}}
               \,\approx\, \frac{1}{n{\rm e}(1-\rho)}
$$
for fixed~$n$ and $\rho\to 1$, so that it can take arbitrarily large values.

\vskip .1in

For the ``standard'' initial conditions $x^{(0)}=(0,\; 0,\;,\dots, 1)$, a very simple lower bound for peak is available for any  feasible $|a|>1,b$.
Indeed, we have $x_k=a^{k-n}$ for $k=n+1,\dots,2n$, so that
$$
\eta(x^{(0)}) \ge x_{2n}=a^n > 1.
$$
Since $a<1+\frac{1}{n}$, this bound is not greater than~$\rm e$.
We omit a more accurate (but much more involved) analysis of the peak. 
For $n=1$ in~\eqref{eq:kuruklis}, it can be shown that  the magnitude of peak can take arbitrarily large values as~$a,b$ approach their upper bounds;
experiments show that this is also true for the general case $n>1$.

\section{Conclusions and future research}

In this paper, attention has been paid to the peaking phenomena in asymptotically stable scalar difference equations with nonzero initial conditions.
It is shown that, in a number of situations, peaks are unavoidable, and several results on the estimates of the magnitude of peak and the peak instant
have been presented. To the best of our knowledge, the very statement of the problems presented here is new.

There are numerous promising directions for future research.

First, it is highly desirable to extend the results obtained to broader classes of difference equations, root locations, and initial conditions.
For instance, of apparent interest are results on peak effects for complex roots; e.g., see Markov's example in Section~\ref{ss:Markov}.
The worst-case conjecture formulated at the end of Section~\ref{SS:lower_upper_bounds} looks natural, 
suggesting that the case of all equal roots implies the
largest deviations. Also, finding worst-case initial conditions in the unit box for classes of stable equations are of interest;
cf. Theorem~\ref{th:bounds_real}.a.

Second, the results can be extended to vector difference equations and the related many-dimensional discrete time dynamical systems.
This also gives raise to the problem of estimating the norms of powers of Schur stable matrices. There are numerous examples of large values
of~$\|A^k\|$ for Schur stable matrices (i.e., for $|\lambda_i|\leq \rho<1$ for all eigenvalues of $A$). 
However there is no systematic theory for lower bounds on powers of stable matrices.

As mentioned in the Introduction, such effects are also typical to the nonasymptotic behavior of some of the modern optimization methods,
and their analysis from the peak effect point of view looks promising.

Third, the behavior of solutions of difference equations with uncertain coefficients (say, those known to lie within given intervals)
is worth analyzing from the peak effect point of view; this relates to robust statements of the problems considered here.

\section*{Funding}


The work of the first two authors was supported by the Russian Science Foundation through grant no. 16-11-10015.


\end{document}